\documentclass[leqno, 12pt]{article}
\usepackage{oxford3}
\usepackage{kokoszka}
\usepackage{graphicx}
\usepackage{amsfonts}
\usepackage{amsmath}
\usepackage{amssymb}
\usepackage{dsfont}
\usepackage{multirow}
\usepackage{xcolor}
\usepackage{enumerate}
\usepackage{multirow}

\renewcommand{\baselinestretch}{1.05}
\newcommand\bzero{{\bf 0}}

\newcommand\bh{{\bf  h}}

\newcommand\bi{{\bf i}}
\newcommand\bj{{\bf j}}
\newcommand\bk{{\bf k}}
\newcommand\bl{{\bf l}}

\newcommand\bn{{\bf n}}

\newcommand\bq{{\bf q}}

\newcommand\bs{{\bf s}}

\newcommand\bt{{\bf t}}

\newcommand\bz{{\bf z}}

\newcommand\mbN{{\mathbb N}}
\newcommand\mbR{{\mathbb R}}

\newcommand\mbZ{{\mathbb Z}}

\newcommand\cF{{\mathcal F}}

\newcommand\cK{{\mathcal K}}
\newcommand\cL{{\mathcal L}}
\newcommand\cN{{\mathcal N}}
\newcommand\cO{{\mathcal O}}

\newcommand\cS{{\mathcal S}}

\newcommand\E{{\mathbb E}}

\renewcommand{\i}{\mathrm{i}}

\newcommand{\norm}[1]{\left\|{#1}\right\|}
\newcommand{\ip}[2]{\left\langle {#1} , {#2} \right\rangle}

\newcommand{\hsnorm}[1]{\left\|{#1}\right\|_{\mathcal{S}}}

\DeclareMathOperator{\krpr}{\widetilde \otimes}
\DeclareMathOperator{\krprt}{\widetilde \otimes_\top}
\DeclareMathOperator*{\tr}{Tr}

\begin{document}

\title{Testing normality of spatially indexed functional data}

\author{Siegfried H{\"o}rmann\\ {\small Technische Universit{\"a}t Graz}
\and Piotr Kokoszka \\ {\small Colorado State University}
\and Thomas Kuenzer\\ {\small Technische Universit{\"a}t Graz}}

\vspace{3cm}

\date{\today}
\maketitle

\begin{abstract}
We develop a test of normality for spatially indexed functions.
The assumption of normality is common in spatial statistics,
yet  no significance tests, or other means of assessment,
have been available for functional data. This paper aims at filling this
gap in the case of functional observations on a spatial grid.
Our test compares the  moments of the spatial (frequency domain)
principal component scores to those  of a suitable Gaussian distribution.
Critical values can be readily obtained from a chi-squared distribution.
We provide rigorous  theoretical justification  for
a broad class of weakly stationary functional random fields.
We perform simulation studies  to assess the the power of the test
against various  alternatives.
An application to Surface Incoming Shortwave Radiation illustrates
the practical value of this procedure.
\end{abstract}

\begin{center}
{\Large
Tester la normalit{\'e} de donn{\'e}es fonctionnelles index{\'e}es spatialement}
\end{center}
\bigskip
{\small
Nous d{\'e}veloppons un test de normalit{\'e} pour les fonctions index{\'e}es
spatialement. Bien que l'hypoth{\`e}se de normalit{\'e} soit courante dans le
domaine de la statistique spatiale, aucun test de significativit{\'e} n'est
disponible dans le cadre de donn{\'e}es fonctionnelles. Ce travail vise {\`a}
combler cette lacune dans le cas de donn{\'e}es fonctionnelles observ{\'e}es
sur une grille spatiale. Notre test compare les moments des scores issus
de l'analyse en composantes principales spatiale (domaine fr{\'e}quentiel) {\`a}
ceux d'une distribution gaussienne appropri{\'e}e. Les valeurs critiques
peuvent {\^e}tre facilement obtenues {\`a} partir d'une distribution du khi-deux.
Un cadre th{\'e}orique rigoureux justifie notre m{\'e}thodologie de test pour une
large classe de champs al{\'e}atoires fonctionnels faiblement stationnaires.
Des simulations sont r{\'e}alis{\'e}es afin d'{\'e}valuer la puissance du test suivant
diff{\'e}rentes hypoth{\'e}ses alternatives.  Une application concernant les
rayonnement {\`a} ondes courtes entrant de surface
 illustre le comportement en pratique de cette proc{\'e}dure.
}

\newpage

\section{Introduction}

Over the last two decades, there has been increasing interest in
functional data, where observations are regarded as elements of a suitable
function space. Several monographs and textbooks give accounts of various
aspects of this field, e.g. \citetext{bosq:2000},
\citetext{ramsay:silverman:2005}, \citetext{shi:choi:2011},
\citetext{HKbook}, \citetext{hsing:eubank:2015} and \citetext{KRbook}.
While studies  of  random samples of functions continue to dominate, the
last decade has seen growing development of functional
data analysis for dependent data. Contributions are becoming more and more
numerous, so we list just a handful of them, without any claim on relative importance.
Within the field  time series analysis,
the emphasis has been on forecasting and inference for temporal dependence, see e.g.
\citetext{hyndman:shang:2009},
\citetext{hormann:kokoszka:2010}, \citetext{liebl:2013},
\citetext{horvath:kokoszka:rice:2014},
\citetext{aue:norinho:hormann:2015}, \citetext{zhang:2016}.
In the ambit of spatial statistics, chief research directions have been
kriging and inference for the spatio--temporal dependence structure. Several
review papers and collections are available, e.g.
\citetext{delicado:2010}, \citetext{mateu:giraldo:2020}
and \citetext{martinez:2020}. Spatial functional data can be regarded as
a type of spatio-temporal data; at each location, we observe a function,
generally defined on a time domain. Temperature or precipitation curves
at spatial locations offer well-known examples, but there are many more.
Most data observed from satellite measurements and outputs of
computer climate models can be treated either as
a temporal sequence of spatial fields or as a field of temporal functions.
We adopt the latter modeling approach.
In either case, the observations are available on a regular spatial grid.

The assumption of Gaussianity has been utilized much more extensively
and profoundly in spatial statistics than in any other field of statistics.
This is chiefly due to the well-established  use of covariance function modeling;
covariances determine distribution only for Gaussian data,
see \citetext{gelfand:schliep:2016} for a broader perspective.
Still, non-Gaussian data seems to be widespread in many contexts.
An overview of non-Gaussianity in climatology was provided by
\citetext{perron:sura:2013}. They investigate
atmospheric key variables observed over several decades and
come to the conclusion that Gaussianity is quite rare in the atmosphere.

Somewhat surprisingly, tests of normality of spatial data have been
absent. Even the application of exploratory tools, like QQ-plots is
questionable because they are justified only if the observations form
a random sample.
\citetext{horvath:kokoszka:wang:2020} derived a normality test for a scalar
spatial field,  which falls to a broad Jarque-Bera family of tests. It is
based on the asymptotic distribution of suitably defined
skewness and kurtosis. In case of spatially dependent data,
these statistics must be defined differently than for random samples.
(Earlier related contributions include \citetext{shenton:bowman:1977},
\citetext{jarque:bera:1980,jarque:bera:1987}.
\citetext{lobato:velasco:2004}, \citetext{bai:ng:2005} and
\citetext{doornik:hansen:2008}.)

This paper is concerned with observations that are functions with domain
$\mathcal{U}$, collected at spatial locations $\bs\in \mathbb{Z}^d$, one function
at each location.
The domain $\mathcal{U}$ can be  a time domain (as in our real data example),
but could also be some other continuous domain, like altitude.
Then $X_\bs(u)$ may be, for example,
the air-temperature at location $\bs$ at time (or altitude) $u$.
No normality tests are currently available for such data, to the best of our knowledge.
Our objective is to fill this gap. The need for such a test arises in many contexts.
For example, \citetext{liu:ray:hooker:2017} developed tests of spatio-temporal
separability and isotropy of spatial functional data, which rely on the assumption
of these data being normal. The same is true of the separability test of
\citetext{constantinou:kr:2017}.
\citetext{gromenko:kokoszka:2017miss} assumed normality to derive
a test for the presence of a common temporal trend in a sample of
spatially indexed functions. Tests of normality of functional random samples
are derived and compared in \citetext{gorecki:hk:2020}, whereas those
for functional time series in \citetext{gorecki:hhk:2018}.

Our approach is based on the decomposition
of a functional spatial field recently derived by
\citetext{kuenzer:hormann:kokoszka:2020}. It uses spatial (frequency domain)
functional principal components analysis (SFPCA) to decompose the functional spatial random field
into $p$ fields of SFPC scores that are orthogonal at all spatial lags.
Under normality, this orthogonality implies that the functional data are actually decomposed
into $p$ layers of independent scalar random fields, each of which is again Gaussian.
Therefore, the testing procedure breaks down
the infinite-dimensional concept of Gaussianity of functions
into testing $p$ independent scalar random fields for Gaussianity.
The number $p$ represents the level of dimension reduction. It is generally a
small, single digit number, often 2, 3 or 4. At each of these levels we
apply a test of normality based on skewness and kurtosis of suitably
defined spatial fields. This is the simplest and most commonly used approach
that turns out to work well.

The remainder of the paper is organized as follows. After presenting
the required background in Section~\ref{s:prelim}, we provide
a self-contained description of our test in Section~\ref{s:derivation}.
Section~\ref{s:aj} is dedicated to its asymptotics justification, with
the proof collected in the appendix. Finite sample performance is investigated
in Section~\ref{s:sim}.
The paper concludes with an application to Surface Incoming Shortwave Radiation
in Section~\ref{s:sis}.

\section{Preliminaries} \label{s:prelim}
Before we formulate the test, we need to introduce the notation
and the framework in which we operate.
We consider functions defined  on a spatial grid in a Euclidean space
of dimension $d$. The functions live in the space  $H= L^2([0,1])$,
the set of square integrable functions on the interval $[0,1]$, with
the usual inner product and norm. The interval $[0,1]$ is considered only
for the convenience of notation; it can be replaced by any other interval.
A functional random field is then an infinite collection
of random functions,  $(X_\bs)_{\bs \in \mbZ^d}$,
where for each  $\bs \in \mbZ^d$, $X_\bs \in H$.
This means that at  each grid point $\bs$ we have a
curve $X_\bs(u)$, $u \in [0,1]$. In most applications,  the variable
$u$ is rescaled time. We assume throughout the paper, that each
function is square integrable, i.e.
\[
E \lnorm X_\bs \rnorm^2 = E \int_0^1 X_\bs^2(u) du < \infty.
\]
Under this assumption, we define a Gaussian functional random field as
follows.

\begin{definition} \label{def:gaussianity}
A functional random field $(X_\bs)_{\bs \in \mbZ^d}$ is called
\emph{Gaussian} if for all $n \in \mbN$,
any deterministic elements $\{v_1, \dots, v_n\} =: V \subset H$ and
any grid points $\{\bs_1, \dots, \bs_n\} =: S \subset \mbZ^d$,
the projections $\ip{X_{\bs_i}}{v_i}$ are normally  distributed, i.e.
\begin{equation}
(\ip{X_{\bs_1}}{v_1},\ldots,\ip{X_{\bs_n}}{v_n})^\top
\sim \cN_n(\mu_{S,V}, \Sigma_{S,V}),
\end{equation}
where $\mu_{S,V}$ and $\Sigma_{S,V}$ depend on the sets $S$ and $V$.
\end{definition}
We note that there are several equivalent definitions of normality
in a Hilbert space. See e.g.\ Chapter~7 of \citetext{laha:roghatgi:1979}.

We now define a stationary functional random field.

\begin{definition} \label{def:stationary}
A functional random field $(X_\bs)_{\bs \in \mbZ^d}$ is called weakly
stationary, if
\begin{enumerate}[(i)]
\item for all~$\bs \in \mbZ^d$, $\E X_\bs = \E X_\bzero$;
\item  for all $\bs, \bh \in \mbZ^d$ and $u, v \in [0,1]$
\[
c_\bh(u,v) := \cov(X_{\bh}(u),X_{\bzero}(v))
= \cov(X_{\bs + \bh}(u),X_{\bs}(v)).
\]
\end{enumerate}
\end{definition}

Observe that the kernel $c_\bh$ is Hilbert--Schmidt, i.e.\
$\iint c_\bh^2(u,v) du dv < \infty$. For any Hilbert--Schmidt kernel $\psi$,
we define the corresponding operator $\Psi$ on $H$ by
$\Psi( y) = \int \psi(u, v) y(v) dv$, $y\in H$.
The integral operator defined by the autocovariance kernel $c_\bh$ is
thus denoted by $C_\bh$.

Next we turn to the concept of SFPCA, which relies on frequency domain concepts. In particular, we  need the so-called spectral density operator. For a weakly stationary functional
random field the integral operator $\cF_{\btheta}^X$ with  the kernel
\begin{equation}\label{d:sd}
f^X_{\btheta} (u, v) := \frac{1}{(2\pi)^d}
\sum_{h \in \mbZ^d} c_\bh(u,v) e^{-\i \bh^\top \btheta}, \qquad
\btheta \in [-\pi,\pi]^d,
\end{equation}
is called the spectral density operator of
$(X_\bs)$ at the spatial frequency $\btheta$.

To ensure convergence of the infinite series in \eqref{d:sd},
we impose the following assumption.

\begin{assumption} \label{ass:setup} The field
$(X_\bs)_{\bs\in\mbZ^d}$ is  weakly stationary
with mean zero and absolutely summable autocovariances in the sense that
\begin{equation} \label{e:summableC}
\sum_{\bh \in \mbZ^d} \tr( C^X_\bh ) < \infty,
\end{equation}
where $\tr( \cdot )$  denotes the trace norm
defined as the sum of the singular values of the operator.
\end{assumption}

Exponentially  decaying  autocovariances satisfy \refeq{summableC},
but it admits  slower decay. Under Assumption~\ref{ass:setup},
the theory of \citetext{kuenzer:hormann:kokoszka:2020} is applicable.
We now outline its  elements we need for the development of the normality test.

The SFPC scores $Y_{m,\bs}$ and the filter functions $\phi_{m,\bl}$
are defined using the eigensystem of the spectral density operator.
Let $\lambda_1(\btheta) > \lambda_2(\btheta) > \dots > 0$
be the ordered eigenvalues of $\cF_{\btheta}^X$ and
$\varphi_m(. | \btheta)$ be the corresponding eigenvectors.
The level $m$ SFPC scores are defined by
\begin{equation} \label{e:scores}
Y_{m,\bs} := \sum_{\bl \in \mbZ^d} \ip{X_{\bs-\bl}}{\phi_{m,\bl}}, \ \ \
\bs \in \mbZ^d,
\end{equation}
where $\phi_{m,\bl}$ is defined by
\begin{equation} \label{e:phiml}
\phi_{m,\bl}(u) := \frac{1}{(2\pi)^d} \hspace{-3pt}
\int\limits_{[-\pi,\pi]^d} \hspace{-7pt}
\varphi_m(u | \btheta) e^{-\i \bl^\top \btheta} d\btheta.
\end{equation}

The SFPC score fields $(Y_{m,\bs}\colon \bs\in\mathbb{Z}^d)$, $1\leq m \leq p$, are then orthogonal
in the sense that
\begin{equation} \label{e:zero-cor}
\cov(  Y_{m,\bs},  Y_{n,\bs^\prime} ) = 0, \ \
{\rm if} \ m \neq n, \ \ \forall\  \bs, \bs^\prime \in \mbZ^d.
\end{equation}

The usual functional principal component (FPC) scores are defined by $\xi_{m, \bs} = \lip X_\bs, v_m \rip$,
where $v_m$ is  the $m$-th FPC. They are uncorrelated at each location,
i.e $\cov(\xi_{m, \bs}, \xi_{n, \bs}) = 0$, if $m\neq n$.
 The analog of \refeq{zero-cor} does
not hold. As we will see in Section~\ref{s:derivation}, it is property \refeq{zero-cor}
that allows us to construct our normality test.
The scores $Y_{m,\bs}$, obtained with the spatial FPCA,
 define a spatial field for each $m$.
They  take into account
data from neighboring spatial locations $\bs$ (in theory all $\bs$)
via a spatial filter $(\phi_{m,\mathbf{l}}(u)\colon \mathbf{l}\in\mathbb{Z}^d)$.
The filters are chosen in such a way  that scores $(Y_{m,\bs}\colon \bs\in\mathbb{Z}^d)$
from different ``layers'' $m$  become mutually orthogonal (uncorrelated).
Under Gaussianity it implies that $(Y_{m,\bs}\colon \bs\in\mathbb{Z}^d)$
and $(Y_{m',\bs}\colon \bs\in\mathbb{Z}^d)$ are independent fields for
$m\neq m'$. One should contrast this property with the
scores in the usual Karhunen-Lo{\'e}ve expansion,
for which we could only conclude that $\xi_{m,\bs}$
is independent of $\xi_{m^\prime,\bs}$ but not necessarily
from $\xi_{m^\prime,\bs^\prime}$ at a different location $\bs^\prime$.

The population scores $Y_{m,\bs}$ can be approximated by their
sample counterparts $\widehat{Y}_{m,\bs}$. The construction of the
sample scores $\widehat{Y}_{m,\bs}$  involves
several steps which are explained in Sections~3.2 and 3.3 of
\citetext{kuenzer:hormann:kokoszka:2020}. In a nutshell,
the spectral density estimator  ${\cF}^X_{\btheta}$
is estimated by a suitably constructed estimator
$\widehat{\cF}^X_{\btheta}$, and the steps listed above are applied
to $\widehat{\cF}^X_{\btheta}$ in place of ${\cF}^X_{\btheta}$, with
infinite sums replaced by truncated sums. The estimated scores are thus defined by
\begin{align}\label{e:Yhat}
\widehat{Y}_{m,\bs}
:= \sum_{\norm{\bl}_\infty \leq L} \ip{X_{\bs-\bl}}{\hat{\phi}_{m,\bl}},
\ \ \ 1 \leq m \leq p,
\end{align}
where $\hat{\phi}_{m,\bl}$ are the Fourier expansion coefficients
of the eigenvectors $\hat \varphi_{m}(\btheta)$ of $\widehat{\cF}_{\btheta}^X$.
While the exact choice of this estimator is not crucial to our method,
we define the estimated spectral density by
\[
\widehat{\cF}_{\btheta}^X := \frac{1}{(2\pi)^r} \sum_{\bh} w_{\bq}(\bh) \; \widehat{C}_\bh \; e^{-\i  \bh^\top \btheta},
\]
where $\widehat{C}_\bh$ are the usual sample autocovariance operators at lag $\bh$,
$w_{\bq}$ is a weight function and $\bq$ is a vector of positive window sizes.
There are different possibilities for the choice of the weight function.
For our calculations, we used the Bartlett kernel
$w_{\bq}(\bz) = \left( 1 - \norm{\bz / \bq} \right)_+$.

\section{Description  of the test} \label{s:derivation}
We will use the sample moments of a generic scalar field
$(Z_\bs)_{\bs \in \mbZ^d}$
observed on a domain  $R_\bn\subset \mbZ^d$ with cardinality $|R_\bn| = N$.
We defined  them by
\[
\hat m_k^Z := \frac{1}{N} \sum_{\bs \in R_\bn} Z_\bs^k.
\]
The Jarque--Bera test compares sample skewness and kurtosis of a distribution
with the corresponding values of a normal distribution.
Let $(Z_\bs)_{\bs\in\mbZ^d}$ be a stationary scalar random field on a grid.
The mean of $Z_\bs$ is denoted by $\mu = \E Z_\bs$ and
the $k$-th central moment is $\mu_k = \E[(Z_\bs - \mu)^k]$.
The variance, as a special case, is denoted by $\sigma^2 = \mu_2$.
Skewness and kurtosis are defined by
\begin{align*}
\tau = \frac{\mu_3}{\sigma^3} \quad \text{ and } \quad \kappa = \frac{\mu_4}{\sigma^4}.
\end{align*}
These parameters can be estimated by the corresponding sample moments.
We call the resulting estimators, $\hat \tau$ and $\hat \kappa$, the
sample skewness and kurtosis.
If the $Z_\bs$ are i.i.d., then the standard Jarque--Bera test
is based on the convergence
\begin{equation} \label{e:jbtest}
\text{JB}_N := N \left( \frac{\hat \tau ^2}{6} + \frac{(\hat \kappa - 3)^2}{24} \right)
   \convd \chi_2^2,
\end{equation}
which holds under the null hypothesis of normality.
Under spatial dependence,  convergence \refeq{jbtest} no longer holds, as explained in
\citetext{horvath:kokoszka:wang:2020}.

For the observed functional field $(X_\bs)_{\bs \in R_\bn}$,  we proceed as follows.
For $1 \le m \le p$, we compute the estimated score fields
$(\widehat{Y}_{m, \bs})$. We center each field and obtain
\[
\widehat Z_{m,\bs} = \widehat{Y}_{m, \bs}  - \hat m^{\widehat{Y}_m}_1,
\ \ \ \ \bs \in R_\bn.
\]
Next, we compute for the levels $m\in\{1,\ldots, p\}$ the statistics related to sample skewness and kurtosis, which
are defined by
\begin{equation} \label{e:SK}
\widehat \cS^{(m)}_\bn := \sqrt{N} \, \hat m_3^{\widehat{Z}_m}, \ \ \
\widehat \cK^{(m)}_\bn  := \sqrt{N}  \left( \hat m_4^{\widehat{Z}_m} - 3
(\hat m_2^{\widehat{Z}_m})^2 \right).
\end{equation}
Finally, we define the test statistic by
\begin{align}\label{e:teststatistics}
\widehat T_p := \sum_{m=1}^p \widehat J_m,
\qquad \text{ with } \qquad
\widehat J_m :=
\frac{ ( \widehat \cS_\bn^{(m)} )^2 }{ 6\, \hat \sigma_{\cS, m}^2 }
+
\frac{ ( \widehat \cK_\bn^{(m)} )^2 }{ 24\, \hat \sigma_{\cK, m}^2 }.
\end{align}
The variance estimators
$\hat \sigma_{\cS, m}^2$ and $\hat \sigma_{\cK, m}^2$
are defined by
\begin{equation}  \label{e:sigmahat}
\hat \sigma_{\cS, m}^2 =
\sum_{\norm{\bl}_\infty \leq L^\prime} \hspace{-5pt} \hat \gamma_{m,\bl}^3,
\quad \text{ and } \quad
\hat \sigma_{\cK, m}^2 =
\sum_{\norm{\bl}_\infty \leq L^\prime} \hspace{-5pt} \hat \gamma_{m,\bl}^4,
\end{equation}
with the sample autocovariances defined by
\[
\hat \gamma_{m, \bh} =
\frac{1}{N} \sum_{\bs\in M_{\bh,\bn}}
 ( \widehat Y_{m,\bs+\bh} - \hat m_1^{Y_m} )
 ( \widehat Y_{m,\bs} - \hat m_1^{Y_m} ).
\]
The set $M_{\bh,\bn}$ is defined as the set of the locations for which
$\bs \in R_\bn$ and $\bs+\bh\in R_\bn$.  The truncation parameter $L^\prime$
is discussed below.
We will see in Section~\ref{s:aj} that  the test statistic $\widehat T_p$ is asymptotically
$\chi^2_{2p}$-distributed under the null.

We conclude this section with an algorithmic description of the test, which
contains guidance on the choice of the tuning parameters and
suitable {\tt R} functions. We use our package {\tt fsd.fd} that is available
on Github. The following recommendations reflect
our experience based on extensive  numerical experiments.

\begin{enumerate}[(1)]
\item In a first step, the spectral density is estimated ({\tt fsd.spectral.density}) on a suitable equidistant integration grid.
The window size parameter $\bq \in \mbN^d$ can
be selected according to a suitable rule of thumb, such as $q_i = \sqrt{n_i}$.
However, we recommend using the automatic, data-driven procedure described
in Section~III of \citetext{kuenzer:hormann:kokoszka:2020}.

\item Using the function {\tt fsd.spca.var}, we obtain an estimate of the
portion of variability that the single layers of SFPC scores explain.
We then choose $p$ such that the first $p$ SFPCs explain
at least 85\% of the total variability of the functional data.

\item For the computation of the filter functions, we use
{\tt fsd.spca.filters} with the parameters {\tt Npc = }$p$ and
The maximum lag {\tt L} to calculate the SFPC filters and scores
is chosen as the smallest integer $L$ such that the filter functions
reach at least 95\% of the total weight,
i.e. $\sum_\bl \norm{\phi_{1,\bl}}^2 \geq 0.95$.

\item We apply the filter functions to the functional random field
({\tt fsd.spca.scores}) to obtain the SFPC scores.

\item Finally, we can use {\tt fsd.jb.test} to conduct the test.
For this, we supply as argument {\tt X.spca} a list with the entry {\tt scores}
and set {\tt var.method = "direct"}.
Regarding the choice of $L^\prime$,
i.e.\ the maximum lag of the SFPC score autocovariances
for the estimation of $\sigma_\cS^2$ and $\sigma_\cK^2$,
we recommend setting the argument {\tt L = }$\norm{\bq}_\infty$
in order to capture enough covariance.
This selection criterion stems from the fact that the
spatial dependence structure of the original data $(X_\bs)$
is mostly reflected in the SFPC scores.
\end{enumerate}

It is possible to shorten the procedure by using the function
{\tt fsd.spca} with suitably preselected arguments and then
follow up with step (5).
The levels of 85\% in step (2) and of 95\% in step (3) are unrelated
and somewhat arbitrary. The first is related
to the variability of the data explained by SFPCs,
the second to the numerical accuracy of the approximation
of the filters.  The level of 85\% in step (2) is fairly standard
in FDA and works well in the context of this paper. The level of 95\% in step (3)
is also typical and turns out to work well in our simulations and applications.

\section{Asymptotic justification} \label{s:aj}
As noted at the end of Section~\ref{s:prelim}, the starting point to the
implementation of the test is an estimator of the spectral density
operator. We need only the following weak assumption, which
is satisfied by the estimator proposed by \citetext{kuenzer:hormann:kokoszka:2020}
for broad classes of functional random fields.

\begin{assumption} \label{ass:Fhat}
The estimator $\widehat{\cF}_{\btheta}^X$ satisfies
\[
\int\limits_{[-\pi,\pi]^d} \hspace{-4pt} \E
\lnorm \widehat{\cF}^X_{\btheta} - {\cF}^X_{\btheta} \rnorm_\cL
 d\btheta \to 0,
 \]
 where $\lnorm \ \cdot \ \rnorm_\cL$ is the operator norm.
\end{assumption}

The next assumption, also used in \citetext{kuenzer:hormann:kokoszka:2020},
is needed to ensure the identifiability and the convergence of the SFPC estimators.

\begin{assumption} \label{ass:alphas}
Let $\alpha_m(\btheta)$ be the spectral gaps, i.e.
\[
\alpha_m(\btheta) := \min\{
\lambda_{m}(\btheta) - \lambda_{m+1}(\btheta) ,
\lambda_{m-1}(\btheta) - \lambda_{m}(\btheta) \}.
\]
We assume that
for all $1 \le m \le p$,
the spectral gaps are bounded from below, such that
\[
\inf \limits_{\btheta \in [-\pi,\pi]^d} \alpha_m(\btheta) =: \beta_m > 0.
\]
\end{assumption}

The final assumption on the population quantities refers to the summability of the
filter functions.

\begin{assumption} \label{a:S3}
For all $1 \le m \le p$, the filter functions of the SFPCs
are absolutely summable in the sense that
$
\sum_{\bl\in\mbZ^d} \norm{\phi_{m,\bl}} < \infty.
$
\end{assumption}

Next we turn to the assumption on the sampling region.

\begin{assumption} \label{a:Rn}
The sampling region is the rectangle
\[
R_\bn = \{\bs \in \mbZ^d:\; 1 \leq s_i \leq n_i \; \forall\, 1\leq i \leq d \}
\]
such that $\min_{1 \le i \le d} n_i \to \infty$.
\end{assumption}

Recall that the number of locations in $R_\bn$ is denoted by $N$, and
the index $\bn$ is used to identify this expanding spatial domain.
Assumption~\ref{a:Rn} could be replaced by a more complex technical assumption,
but it simplifies arguments  and is generally satisfied
in applications.

Asymptotic results are stated in terms of the following
quantities:
\begin{equation} \label{e:GH}
\quad G(N)  := \hspace{-5pt} \int\limits_{[-\pi,\pi]^d} \hspace{-5pt}
 \lnorm \widehat{\cF}^X_{\btheta} - {\cF}^X_{\btheta} \rnorm_\cL \hspace{-5pt}
 d\btheta, \ \ \ \
H_m(L)  := \Bigg( \sum_{\norm{\bl}_\infty > L} \hspace{-5pt}
 \norm{\phi_{m,\bl}}^2 \Bigg)^{1/4}.
\end{equation}

Our main asymptotic result, Theorem~\ref{t:main},
 states that the asymptotic null distribution
of the test statistic $\widehat T_p$  defined by \refeq{teststatistics} is chi-square
with $2p$ degrees of freedom. Recall that $p$ is the the number of levels
in the SFPCA used to construct the statistic. At each level, the asymptotic
distribution is chi-square with two degrees of freedom, and by utilizing
the asymptotic independence between the levels, we obtain the desired result.
The proof is presented in the appendix. It is quite complex because
independence properties hold only at the population level. At the sample level,
independence is only asymptotic.

To understand Theorem~\ref{t:main}, we first consider analogs of the statistics
$\widehat \cS^{(m)}_\bn$ and $\widehat \cK^{(m)}_\bn$ defined in
\refeq{SK} in terms of the population scores, $Y_{m, \bs}$,  rather than
the estimated scores $\widehat{Y}_{m, \bs}$. We thus set
\begin{equation} \label{e:SK-pop}
\cS^{(m)}_\bn := \sqrt{N} \, \hat m_3^{Z_m}, \ \ \
\cK^{(m)}_\bn  := \sqrt{N}  \left( \hat m_4^{Z_m} - 3
(\hat m_2^{Z_m})^2 \right), \ \ \ Z_{m, \bs} = Y_{m, \bs} - m_1^{Y_m}.
\end{equation}
Set
\[
\ga_{m, \bh} = \cov(Y_{m, \bs+ \bh}, Y_{m, \bs }).
\]
By Lemma \ref{lemma:convergence1},
\begin{equation} \label{e:skn}
\begin{pmatrix}
\cS_\bn^{(m)} \\
\cK_\bn^{(m)}
\end{pmatrix}
 \convd
N_2\left( \begin{pmatrix}0\\0\end{pmatrix},
\begin{pmatrix}
6\, \sigma_{\cS, m}^2 & 0 \\
0 & 24\, \sigma_{\cK, m}^2
\end{pmatrix} \right),
\end{equation}
where
\begin{equation}\label{e:skv}
\sigma_{\cS, m}^2 = \sum_{\bh\in\mbZ^d} \gamma_{m,\bh}^3,
\quad \text{ and } \quad
\sigma_{\cK, m}^2 = \sum_{\bs\in\mbZ^d} \gamma_{m, \bh}^4.
\end{equation}
Theorem \ref{t:main} states that under technical assumptions, the null
distribution is $\chi^2_{2p}$, as long as the asymptotic variances
$\sigma_{\cS, m}^2$ and   $\sigma_{\cK, m}^2$ can be consistently estimated.

\begin{theorem} \label{t:main}
Suppose Assumptions~\ref{ass:setup}, \ref{ass:Fhat}, \ref{ass:alphas}, \ref{a:S3}
and \ref{a:Rn}
are satisfied and $(X_\bs)$ is a Gaussian process.
Suppose  that $L = L(N) \to \infty$ such that $L^d \, G(N) \convP 0$,
and $L^\prime = L^\prime(N) \to \infty$,
and for $1 \le m \le p$,
$\hat \sigma_{\cS, m}^2$ and $\hat \sigma_{\cK, m}^2$ in \refeq{teststatistics}
are consistent estimators of $\sigma_{\cS, m}^2$ and
$\sigma_{\cK, m}^2$ in  \refeq{skv}. Then $\widehat T_p \convd \chi^2_{2p}$.
\end{theorem}

The remaining question that must be addressed is if the estimators
considered in Section~\ref{s:derivation} are consistent. This is indeed the
case, as stated in the following proposition.

\begin{proposition} \label{p:con}
Suppose Assumptions~\ref{ass:setup}, \ref{ass:Fhat}, \ref{ass:alphas}
and \ref{a:Rn}
hold and  $(X_\bs)$ is a Gaussian process.
 Assume that, as  $N\to \infty$,
\[
L, \, L^\prime \to \infty, \ \ \ \text{ such that} \ \ \
L^\prime = o(\min_{1 \le i \le d} n_i),
\ \ \  (L^{\prime})^{d/3} H_m(L) \convP 0.
\]
Then,  $\hat \sigma_{\cS,m} ^2$ and $\hat \sigma_{\cK, m} ^2$
defined by \refeq{sigmahat}
are  consistent for the asymptotic variances \refeq{skv}.
\end{proposition}

Proofs of both Theorem~\ref{t:main} and Proposition~\ref{p:con}
are developed in the appendix.

\section{Finite sample performance}\label{s:sim}
The purpose of this section is to assess the performance of our test
by means of a simulation study.  To evaluate empirical size and power,
we simulated samples
\[
X_{s,t}(u), \quad s, t \in \{ 1, 2, \dots, n \}, \; u \in [0,1],
\]
according to the following autoregressive scheme:
\[
X_{s,t} = A\, X_{s-1,t} + B\, X_{s,t-1} + \varepsilon_{s,t},
\]
where the $\varepsilon_{s,t}$ are  i.i.d. errors  and
$A$ and $B$ are two operators. We explain the details below, but the
idea is that the $\varepsilon_{s,t}$ are Gaussian curves under the null hypothesis
and have  different distributions under alternatives.

We simulated samples in the finite dimensional space
spanned by 15 Fourier basis functions $(v_i)_{1\leq i \leq 15}$.
The $X_\bs$ are determined by their coefficient vectors.
The operators are represented by $15 \times 15$ dimensional coefficient
matrices whose $(i,j)$--entries are simulated as independent normal
random variables with mean $0$ and variance $(i^2+j^2)^{-1/2}$.
The operators are then scaled to the operator norm
$\| A \|_\cL  \hspace{-4pt} = 0.6$ and $\|B\|_\cL \hspace{-4pt} = 0.35$,
which ensures convergence and stationarity
\cite{kuenzer:hormann:kokoszka:2020}.

Under the null, we simulate $\varepsilon_{s,t}$ such that the Fourier coefficients
$\epsilon_{s,t,i} = \ip{ \varepsilon_{s,t} }{v_i}$ are independent, normal
with mean zero and variance $2^{-i}$.
For these data, three SFPCs explain 85\% of the variance.
Under the alternative,
we simulate the $\epsilon_{s,t,i}$ from Johnson's $S_U$ distribution
\citetext{johnson:1949},
which is a leptokurtic distribution family
with four parameters that allows to fix the first four moments.
It is defined as a transformation of the normal distribution by
\[
X = \xi + \lambda \sinh\left( \frac{Z - \gamma}{\delta} \right),
\quad \text{ with } \quad
Z \sim N(0,1).
\]
For this distribution, all moments exist.
We choose the parameters such that mean and variance of $\epsilon_{s,t,i}$
are the same as under the null setting, and specify
skewness $\tau$ and kurtosis $\kappa$ separately.
We denote such a distribution as $S_U(\tau, \kappa)$.

We simulated these data generating processes for several sample sizes.
Empirical rejection rates can be found in Table~\ref{tab:newsims}.
The critical values are
the quantiles of the $\chi_{2p}^2$-distribution.

\begin{table}[t]
\centering
\begin{footnotesize}
\begin{tabular}{lrrrrrrrrrrrrrrr}
$N$ & & \multicolumn{4}{c}{$12 \times 12$} & & \multicolumn{4}{c}{$25 \times 25$} & & \multicolumn{4}{c}{$50 \times 50$} \\
\cline{3-6}\cline{8-11}\cline{13-16}
\multicolumn{2}{l}{$p$} &    1 &    2 &    3 &    4 &  &    1 &    2 &    3 &    4 &  &    1 &    2 &    3 &    4 \\
  \hline
Gaussian             & & 0.05 & 0.06 & 0.06 & 0.07 & & 0.05 & 0.05 & 0.06 & 0.06 & & 0.06 & 0.06 & 0.06 & 0.06 \\
\hline
$S_U(0,\, 3.2)$      & & 0.08 & 0.10 & 0.11 & 0.12 & & 0.10 & 0.12 & 0.15 & 0.17 & & 0.15 & 0.19 & 0.26 & 0.35 \\
$S_U(0.1,\, 3.2)$ \hspace{-1cm}   & & 0.08 & 0.09 & 0.10 & 0.11 & & 0.12 & 0.15 & 0.18 & 0.22 & & 0.24 & 0.30 & 0.45 & 0.58 \\
$S_U(0,\, 3.5)$      & & 0.12 & 0.14 & 0.17 & 0.18 & & 0.20 & 0.27 & 0.35 & 0.44 & & 0.45 & 0.59 & 0.78 & 0.91 \\
$S_U(0.25,\, 3.5)$ \hspace{-1cm}  & & 0.14 & 0.17 & 0.20 & 0.22 & & 0.32 & 0.42 & 0.56 & 0.66 & & 0.70 & 0.85 & 0.96 & 0.99 \\
$S_U(0,\, 4)$        & & 0.19 & 0.25 & 0.29 & 0.32 & & 0.42 & 0.55 & 0.69 & 0.81 & & 0.78 & 0.90 & 0.98 & 1.00 \\
$S_U(0.5,\, 4)$      & & 0.26 & 0.32 & 0.39 & 0.45 & & 0.63 & 0.80 & 0.91 & 0.97 & & 0.94 & 0.99 & 1.00 & 1.00 \\
$S_U(0,\, 5)$        & & 0.30 & 0.39 & 0.48 & 0.55 & & 0.66 & 0.83 & 0.94 & 0.98 & & 0.97 & 0.99 & 1.00 & 1.00 \\
$S_U(1,\, 5)$        & & 0.51 & 0.62 & 0.72 & 0.79 & & 0.91 & 0.99 & 1.00 & 1.00 & & 0.99 & 1.00 & 1.00 & 1.00 \\
$S_U(0,\, 6)$        & & 0.41 & 0.52 & 0.60 & 0.67 & & 0.81 & 0.92 & 0.98 & 1.00 & & 0.99 & 1.00 & 1.00 & 1.00 \\
$S_U(1,\, 6)$        & & 0.55 & 0.65 & 0.76 & 0.81 & & 0.93 & 0.98 & 1.00 & 1.00 & & 1.00 & 1.00 & 1.00 & 1.00 \\
  \hline
\end{tabular}
\end{footnotesize}
\noindent\caption{Empirical rejection rates with nominal size $\alpha = 0.05$
based on $3000$ replications.}\label{tab:newsims}
\end{table}

The nominal size is attained in the simulation settings we considered.
The empirical size is stable with respect to $p$,
while the empirical power increases with $p$.  (The number of SFPCs we would use according to our goal
of explaining 85\% of the variability in the data is $p=3$.)
This can be explained by the simulation setting we have chosen,
where all principal components are similarly non-Gaussian.
For other settings, it may well be possible that
the deviation from Gaussianity only manifests
in certain principal components.
In such cases, choosing a large $p$ might drown out this signal
and, in fact, lower the power of the test.

Our results suggest that
the power of the test is similar to the functional time series case
considered in  \citetext{gorecki:hhk:2018} and the
scalar spatial case studied  in \citetext{horvath:kokoszka:wang:2020}.

Our conclusion is that the test is able to reliably detect
moderate deviations from normality even
when the sample size is small ($12 \times 12$).
Starting from medium sample sizes ($25 \times 25$),
most practically relevant departures from normality
can be detected. The case closest to Gaussianity that we considered
is a kurtosis of $\kappa = 3.2$, which is the kurtosis of a
Student-$t$-distribution with $34$ degrees of freedom.
The $S_U$-distribution with this kurtosis is almost identical to
the corresponding Student-$t$-distribution
except for the tails that permit all moments to be finite.
In scalar data, this kind of distribution is visually indistinguishable from
the normal distribution. For many practical applications,
this deviation from normality can even be neglected.

\section{Application to Surface Incoming Shortwave Radiation} \label{s:sis}

\begin{figure}
\begin{center}
\includegraphics[width=0.47\textwidth, trim={0 10 0 4mm}, clip]
{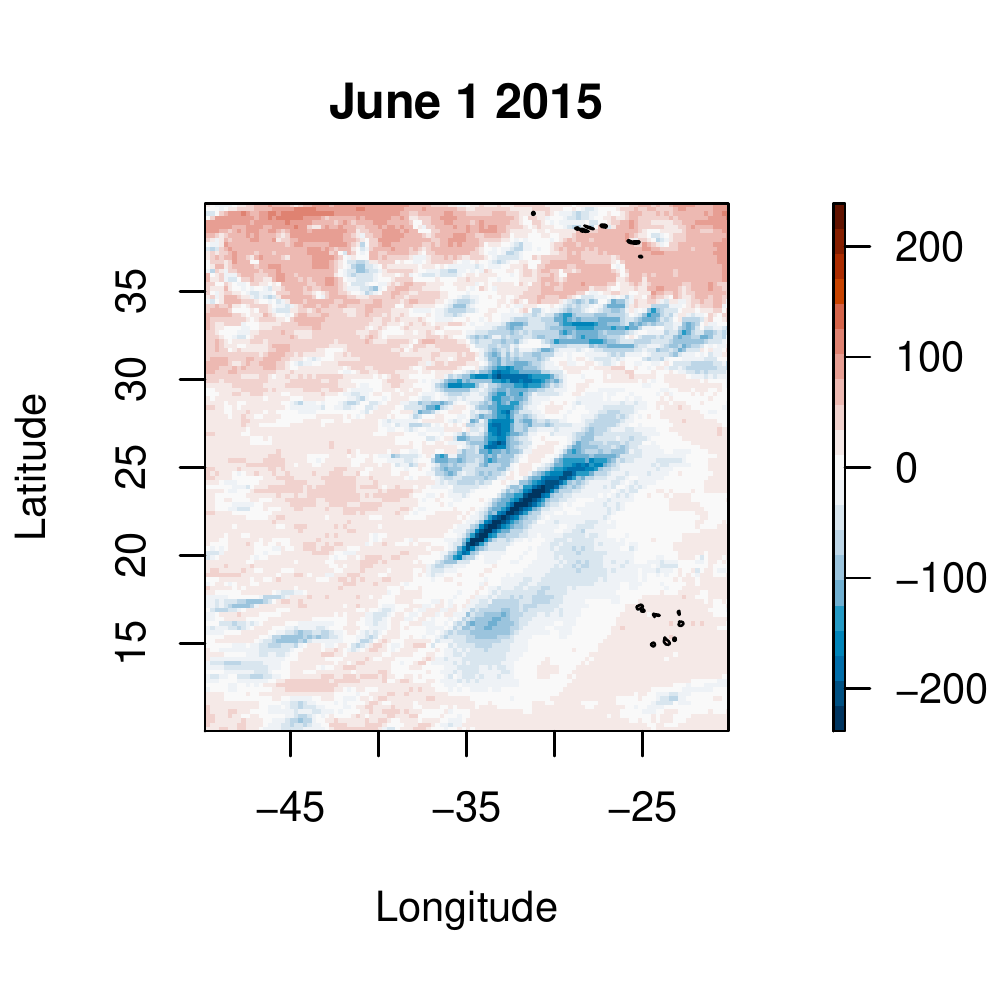}%
\hfill%
\includegraphics[width=0.47\textwidth, trim={0 10 0 4mm}, clip]
{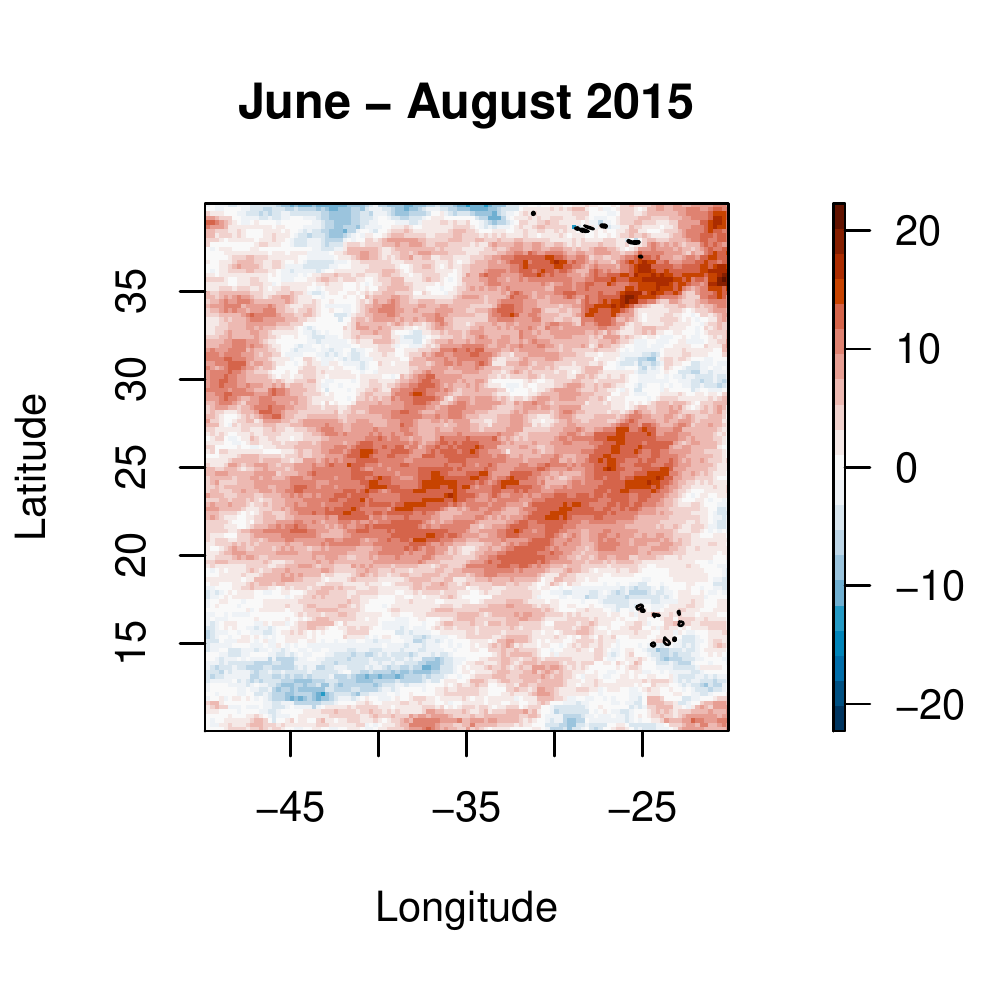}%
\vspace{-8mm}
\end{center}
\caption{The SIS radiation anomalies as viewed over different time scales.
Daily mean (left) vs 3-months mean (right).
Note that the regional mean of the SIS radiation anomaly is positive in 2015.
For orientation, the archipelagos of the Azores and Cape Verde
are added to the map.
\label{fig:sis1}}
\end{figure}

Research by \citetext{perron:sura:2013} on many aspects of the atmosphere suggests
that when taking daily mean values, many variables do not follow a Gaussian
distribution.
\citetext{horvath:kokoszka:wang:2020} analyzed monthly mean sea surface temperature data
from various regions of the world and came to the conclusion that one needs to be
cautious when assuming Gaussianity of spatial data.
We want to explore what the application of  our test to an important data set reveals.
We study  Surface Incoming Shortwave (SIS) radiation data,
also called solar surface irradiance,
provided by EUMETSAT \cite{eumetsat:2009}.
These data are available for free download at \texttt{https://www.cmsaf.eu/}.
The SIS data measure how much solar radiation reaches a certain point of the earth's surface on average on a given day. It is measured in W/m$^2$.
The SIS radiation has as a natural upper limit the solar radiation
that reaches the top of the atmosphere,
which varies naturally due to the angle of the sun at a certain latitude
on a given day of the year.
The amount of radiation that reaches the surface is then reduced
by the absorption in the atmosphere, mainly due to clouds.
For our purposes, we extracted a region in the Northern Atlantic ocean
(10$^\circ$--40$^\circ$N, 50$^\circ$--20$^\circ$E)
at a spatial resolution of 0.25$^\circ$.
This yields 120 $\times$ 120 spatial measurement points with a temporal time span of 33 years (1983 to 2015) of daily observations.
The region we chose has a homogeneous surface with little topographical features.
Also, it doesn't comprise high latitudes which would induce distortion
due to the projection onto the angular grid. Hence,
we can assume that after suitably demeaning the data,
the resulting field is stationary.

We take the daily measurements of the SIS radiation data for the summer months June, July and August (92 days),
and fit smooth curves for each year and point on the grid.
We use a 30-dimensional B-spline base, notably smoothing out
the rather volatile data.
The resulting curves can be interpreted as a moving average over time
that indicates how much solar radiation reaches the surface
around a given point in time.
Because we have 33 years of data, we center the data by subtracting the long-term mean from each point on the grid and day of the year.
The centered data can be referred to as anomalies.
Snapshots of the resulting data can be seen
in Figures~\ref{fig:sis1} and \ref{fig:sis2}.

\begin{figure}
\begin{center}
\includegraphics[width=0.8\textwidth, trim={0 10 0 14mm}, clip]
{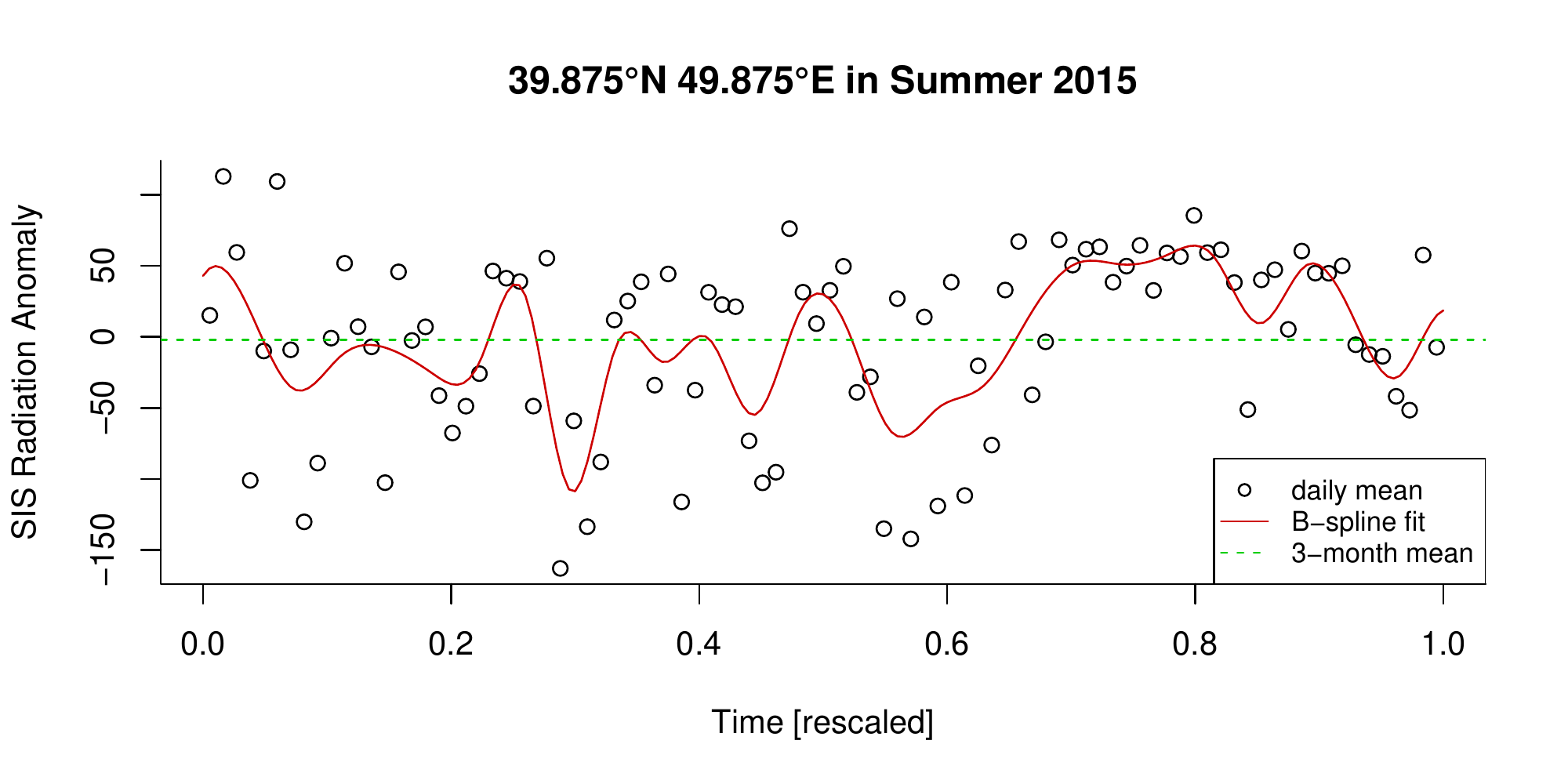}%
\vspace{-8mm}
\end{center}
\caption{The SIS radiation anomalies at 39.875$^\circ$N 49.875$^\circ$E in summer 2015.
The time (June--August) is rescaled to $[0,1]$.
\label{fig:sis2}}
\end{figure}

From the nature of the daily mean data
and from the left plot in Figure~\ref{fig:sis1},
one can already guess that it will probably exhibit notable skewness,
which excludes Gaussianity.
But if we average over a longer time period, say the summer season,
it is conceivable that
the mean values approach a Gaussian distribution.
While the data may still deviate from normality,
it will be close enough to justify the application of algorithms that require Gaussianity.
We will illustrate this process of smoothing out the non-Gaussianity
by employing the scalar test by \citetext{horvath:kokoszka:wang:2020} to verify that
the marginal distributions (data on a day-by-day level) clearly do not follow a Gaussian distribution,
while the mean values from a summer season do.
These are the two extrema of smoothing.
The smoothed B-spline curves will naturally lie somewhere in between.
In a strict sense, the curves will not be Gaussian,
but our test helps assessing whether the assumption of Gaussianity
is reasonable or not.

In the following, we use  $L^\prime = 20$
to capture enough of the covariance of the SFPC scores.
We set $q_i = 15$ and follow the other steps of the algorithm
at the end of Section~\ref{s:derivation}.
For comparison, we use the test of \citetext{horvath:kokoszka:wang:2020}
with the power estimator of the covariances.
In order to guarantee comparability,
we use the same truncation parameter as for our functional test.

\begin{table}
\centering
\begin{small}
\begin{tabular}{rrr|rrr|rrr}
  \hline
 year & B-spline & 3-m mean &  year & B-spline & 3-m mean &  year & B-spline & 3-m mean \\
  \hline
 1983 & \textit{0.0132} &         0.1190    & 1994 &         0.7971  &         0.8681    & 2005 &         0.4951  & \textit{0.0437} \\
 1984 & \textit{0.0455} &         0.1783    & 1995 & \textit{0.0000} &         0.3721    & 2006 & \textit{0.0010} & \textit{0.0100} \\
 1985 & \textit{0.0000} &         0.4274    & 1996 & \textit{0.0498} &         0.3472    & 2007 & \textit{0.0010} &         0.8964  \\
 1986 &         0.5611  &         0.0989    & 1997 & \textit{0.0002} &         0.1038    & 2008 &         0.1465  &         0.5577  \\
 1987 & \textit{0.0003} &         0.0828    & 1998 & \textit{0.0193} &         0.7030    & 2009 & \textit{0.0005} &         0.2197  \\
 1988 & \textit{0.0001} & \textit{0.0039}   & 1999 & \textit{0.0000} & \textit{0.0281}   & 2010 & \textit{0.0000} &         0.4231  \\
 1989 & \textit{0.0000} &         0.9229    & 2000 & \textit{0.0000} &         0.6011    & 2011 &         0.6868  &         0.1196  \\
 1990 & \textit{0.0000} &         0.1030    & 2001 & \textit{0.0083} &         0.9889    & 2012 & \textit{0.0425} &         0.1164  \\
 1991 & \textit{0.0000} & \textit{0.0002}   & 2002 & \textit{0.0000} &         0.3769    & 2013 & \textit{0.0018} &         0.1620  \\
 1992 & \textit{0.0322} &         0.1509    & 2003 &         0.8110  &         0.2778    & 2014 & \textit{0.0000} &         0.2479  \\
 1993 & \textit{0.0000} & \textit{0.0000}   & 2004 & \textit{0.0000} &         0.2282    & 2015 &         0.0971  &         0.4810  \\
  \hline
\end{tabular}
\end{small}
\noindent\caption{$P$-values of the  normality test
applied to the SIS radiation anomaly curves on a $120 \times 120$ grid.
The number $p$ of SFPCs is different
in each year sample, depending on the fraction  of variance explained
by the SFPCs.
Rejections of the null hypothesis on a level $\alpha=0.05$ are marked in italics.
\label{tab:sispvals}}
\end{table}

Results are presented in Table~\ref{tab:sispvals}.
We do not report the results of the univariate tests on the single days,
as aggregation of the marginal $p$-values is a delicate matter
and it suffices to say that most $p$-values were in the range of $<10^{-4}$,
indicating a strong departure from Gaussianity.
We note that for the samples of the 3-months mean SIS radiation
and a significance level of $\alpha = 0.05$,
the test by \citetext{horvath:kokoszka:wang:2020} rejects the null hypothesis of a Gaussian distribution in merely 6 of the 33 years,
marked by the $p$-values in italics.
This means that for most years,
the 3-months mean of the summer months can actually be assumed to be Gaussian,
which is a consequence of a general CLT principle.
The number  of rejections rises to 26 if we use the fitted B-spline curves.
The most general conclusion is that the fitted curve data
cannot be assumed to be  Gaussian for most years.
This means that by smoothing out the daily fluctuations,
we do not approach the Gaussian distribution as fast as expected.
However, this is subject to changes
depending on the amount of smoothing that is applied on the data.
After sufficient smoothing of the daily data,
the application of tools of spatio-temporal
statistics which are based on the assumption of Gaussianity might be justified.

Depending on the year, our test uses  11--13 SFPCs to cover 85\% of the variance.
For some of the years in which we don't obtain a rejection, further
inspection shows that strong evidence of non-Gaussianity appears
only in the 15th or 16th SFPC.
Naturally, this is not captured by our test.
This problem would not occur  in multivariate data,
where the same procedure can be applied with the number of
principal components set to the dimension of the multivariate data.


\bigskip

\bibliographystyle{oxford3}
\renewcommand{\baselinestretch}{0.9}

\bibliography{tho}

\bigskip

\noindent{\bf Acknowledgements} \
Research partially supported by NSF grants  DMS--1914882 and DMS--1923142.

\bigskip \bigskip

\centerline{\Large \bf Appendix}

\appendix

\section{Proofs of Theorem \ref{t:main} and Proposition \ref{p:con} }\label{s:p}
We begin with two known  results, which play important role in our arguments.
Theorem \ref{thm:isserlis} is well--known.  Theorem \ref{thm:arcones} was essentially
established by \citetext{arcones:1994}. Although the title of the paper suggests otherwise,
the results are also true for Gaussian fields,  not only for sequences.
Theorem \ref{thm:arcones} is thus the spatial version of Theorem~2
in \citetext{arcones:1994}.

\begin{theorem}[Isserlis]\label{thm:isserlis}
If $(W_1, \dots, W_n)$ is a multivariate normal vector with zero mean,
then
$\E\left[ \prod W_i \right] = 0$ if $n$ is odd
and
$\E\left[ \prod W_i \right] = \sum \prod \E[W_i W_j]$ if $n$ is even,
where the sum is over all partitions of $\{1, \dots, n\}$
into pairs $\{i, j\}$.
\end{theorem}

\begin{theorem}[Arcones]\label{thm:arcones}
Let $(X_\bs)_{\bs \in \mbZ^d}$ be a stationary Gaussian random field of
mean-zero $\mbR^r$-valued vectors denoted by
$X_{\bs} = (X^{(1)}_{\bs}, \dots, X^{(r)}_{\bs})^\top$,
and $f: \mbR^d \to \mbR$ a function with $\E[f^2(X_1)] < \infty$.

Define
\[
r^{(p,q)} (\bk) = \E\left[ X^{(p)}_{\bs} X^{(q)}_{\bs+\bk} \right]
\]
for $k \in \mbZ^d$ and $1 \leq p, q \leq r$.
Suppose that
\begin{align*}
\lim_{N\to\infty} \frac{1}{N}
\sum_{\bk \in R_\bn} \sum_{\bl \in R_\bn}
r^{(p,q)} (\bk -\bl)
& \qquad \text{ and } &
\lim_{N\to\infty} \frac{1}{N}
\sum_{\bk \in R_\bn} \sum_{\bl \in R_\bn}
\left( r^{(p,q)} (\bk -\bl) \right)^2
&
\end{align*}
exist for $1 \leq p, q \leq r$.
Then
\begin{equation}
\frac{1}{\sqrt N} \sum_{\bs \in R_\bn}
\left( f(X_\bs) - \E[f(X_\bs)] \right)
\convd
N(0, \sigma^2),
\end{equation}
where
\begin{equation}
\sigma^2 := \sum_{\bh \in \mbZ^d}
\cov\left( f(X_\bh), f(X_\bzero) \right) .
\end{equation}
Moreover, there exists a constant $c$
depending only on the field of covariances
such that
\begin{equation}
\E \left(
\frac{1}{\sqrt N} \sum_{\bs \in R_\bn}
\left( f(X_\bs) - \E[f(X_\bs)] \right)
\right)^2 \leq c \;
\var ( f(X_\bzero) )
\end{equation}
for each $\bn$ and function $f$ with finite second moment.
\end{theorem}

\bigskip

We define the truncated version of the population scores in \refeq{scores} by
\[
\widetilde{Y}_{m,\bs} = \sum_{\| \bl \|\le L} \ip{X_{\bs-\bl}}{\phi_{m,\bl}}, \ \ \
\bs \in \mbZ^d.
\]
If $(X_\bs)_{\bs \in \mbZ^d}$ is a Gaussian process,
then $(Y_{m,\bs})_{\bs \in \mbZ^d}$ and
$(\widetilde{Y}_{m,\bs})_{\bs \in \mbZ^d}$ are also Gaussian processes.
This is not the case for the $\widehat{Y}_{m,\bs}$,
as the filters $\hat{\phi}_{m,\bl}$ are estimated,
so the $\widehat{Y}_{m,\bs}$ are nonlinear functions of the  observations $X_\bs$.

In the following, we assume that the value of  $p$ has been selected.
To simplify the formulas,
from now on, we will assume that the data are centered, i.e. $\E X_\bs = 0$.
This has no impact on the arguments.  To further lighten the notation,
we omit the index $m$ in the subscripts
of the objects related to the level $m$ SFPCs and
simply write $Y_{\bs}$ for $Y_{m,\bs}$, $\phi_{\bl}$ for $\phi_{m,\bl}$, and so on.
We will  use
\begin{align*}
\cS_\bn & := \sqrt{N} \, \hat m_3^{Z}, &
\cK_\bn & := \sqrt{N}  \left( \hat m_4^{Z} - 3 (\hat m_2^{Z})^2 \right), &
\text{ with } Z_\bs & := Y_\bs - \hat m^{Y}_1, \\
\widetilde \cS_\bn & := \sqrt{N} \, \hat m_3^{\widetilde Z}, &
\widetilde \cK_\bn & := \sqrt{N}  \left( \hat m_4^{\widetilde Z} - 3 (\hat m_2^{\widetilde Z})^2 \right), &
\text{ with } \widetilde Z_\bs & := \widetilde Y_\bs - \hat m^{\widetilde Y}_1.
\end{align*}

Our first lemma establishes  asymptotic distributions of
population quantities, which approximate  the
the corresponding statistics at a  level $m$.

\begin{lemma} \label{lemma:convergence1}
Suppose $(X_\bs)_{\bs\in\mbZ^d}$ is a Gaussian process satisfying
Assumptions~\ref{ass:setup} and \ref{a:S3}
and $\min\limits_{1\leq i \leq d} n_i \to \infty$,
then the sums
\begin{align}
\frac{1}{\sqrt N} \sum_{\bs \in R_\bn} Y_{\bs}, &&
\frac{1}{\sqrt N} \sum_{\bs \in R_\bn} (Y_{\bs}^2 - \E[Y_\bs^2]), &&
\frac{1}{\sqrt N} \sum_{\bs \in R_\bn} Y_{\bs}^3, &&
\frac{1}{\sqrt N} \sum_{\bs \in R_\bn} (Y_{\bs}^4 - \E[Y_\bs^4])
\label{e:sumsnormal}
\end{align}
are asymptotically normal
and
\begin{equation} \label{e:sknormal}
\begin{pmatrix}
\cS_\bn \\
\cK_\bn
\end{pmatrix}
 \convd
N_2\left( \begin{pmatrix}0\\0\end{pmatrix},
\begin{pmatrix}
6\, \sigma_\cS^2 & 0 \\
0 & 24\, \sigma_\cK^2
\end{pmatrix} \right),
\end{equation}
where
\begin{equation} \label{e:skvars}
\sigma_\cS^2 = \sum_{\bh\in\mbZ^d} \gamma_\bh^3,
\quad \text{ and } \quad
\sigma_\cK^2 = \sum_{\bs\in\mbZ^d} \gamma_\bh^4.
\end{equation}
Furthermore, it holds for $k \leq 4$ that
\begin{equation} \label{e:momentconvergence}
\frac{1}{\sqrt N} \; \sum_{\bs \in R_\bn} ( Y_\bs^k - \widetilde Y_\bs^k - \E[Y_\bs^k - \widetilde Y_\bs^k] ) \convP 0, \qquad \text{ if } L \xrightarrow{N\to\infty} \infty.
\end{equation}
\end{lemma}
\begin{proof}
We prove this Lemma by using Theorem~\ref{thm:arcones}.
Its conditions can be  easily verified,  as they mostly concern the summability
of covariances, which follows from our assumptions.
The asymptotic normality of the sums in \refeq{sumsnormal}
is a simple application of Theorem~\ref{thm:arcones}.
Let us denote $\gamma = \E[Y_\bzero^2]$.
Gaussianity implies
$\E[Y_\bs^3] = 0$ and
$\E[Y_\bs^4] = 3\, \gamma^2$.
As for statement \eqref{e:sknormal},
we first show that
\[
\sqrt N \begin{pmatrix}
\hat m^Y_3 - 3\, \gamma\, \hat m^Y_1 \\
\hat m^Y_4 - 6\, \gamma\, \hat m^Y_2 + 3\, \gamma^2
\end{pmatrix}
 \convd
N_2\left( \begin{pmatrix}0\\0\end{pmatrix},
\begin{pmatrix}
6\, \sigma_\cS^2 & 0 \\
0 & 24\, \sigma_\cK^2
\end{pmatrix} \right)
\]
by employing the Cram\'er--Wold device as follows.
Let $\lambda_1, \lambda_2 \in \mbR$,
and again using Theorem~\ref{thm:arcones},
we have
\[
\frac{1}{\sqrt{N}} \sum_{\bs \in R_\bn}
\left( \lambda_1 \; (Y_{\bs}^3 - 3\, \gamma\, Y_{\bs})
 + \lambda_2 \, ( Y_{\bs}^4 - 6\, \gamma\, Y_{\bs}^2 + 3\, \gamma^2 ) \right)
\convd N(0, \tau^2),
\]
where
\begin{align*}
\tau^2 = \sum_{\bh \in \mbZ^d} \textstyle
\E\big[ &
\left( \lambda_1 \; ( Y_{\bzero}^3 - 3\, \gamma\, Y_{\bzero} )
 + \lambda_2 \, ( Y_{\bzero}^4 - 6\, \gamma\, Y_{\bzero}^2 + 3\, \gamma^2 ) \right)
\\ & \cdot
\left( \lambda_1 \; ( Y_{\bh}^3 - 3\, \gamma\, Y_{\bh} )
 + \lambda_2 \, ( Y_{\bh}^4 - 6\, \gamma\, Y_{\bh}^2 + 3\, \gamma^2 ) \right)
\big].
\end{align*}
Under Gaussianity of $(X_\bs)$, $(Y_\bs)$ is also a Gaussian process and
the cross-covariance between the two summands is evidently zero.
Isserlis's theorem (Theorem~\ref{thm:isserlis})
then yields the simple forms
of $6\, \sigma_\cS^2$ and $24\, \sigma_\cK^2$,
since all the lower order terms cancel.

Now we can see from \refeq{sumsnormal} that
$\hat m_k^Y$, for $k \leq 4$, are $\cO_P(N^{-1/2})$.
Expanding the sums, we can deduce that
\begin{align*}
\cS_\bn & =
\sqrt{N}\, \hat m^Z_3 =
\sqrt{N}\, ( \hat m^Y_3 - 3 \gamma\, \hat m^Y_1)
 + \cO_P(N^{-1/2}), \\
\cK_\bn & =
\sqrt{N} \left( \hat m^Z_4 - 3(\hat m^Z_2)^2 \right) =
\sqrt{N} \left( \hat m^Y_4 - 6\, \gamma\, \hat m^Y_2 + 3\, \gamma^2 \right)
 + \cO_P(N^{-1/2}),
\end{align*}
which amounts to convergence in probability.
This proves \refeq{sknormal}.

We prove \refeq{momentconvergence} by using the last statement of Theorem~\ref{thm:arcones},
which yields
that there exists a constant $c$ such that
\[
\frac{1}{N} \; \E \left( \sum_{\bs \in R_\bn} ( Y_\bs^k - \widetilde Y_\bs^k - \E[Y_\bs^k - \widetilde Y_\bs^k] )  \right)^2
 \leq
c \; \var(Y_\bzero^k - \widetilde Y_\bzero^k)
\]
This constant $c$ does not depend on $N$,
but only on the sequence of covariances
$\cov(Y_\bh^k - \widetilde Y_\bh^k, Y_\bzero^k - \widetilde Y_\bzero^k)$.
What is still left to show is that it is possible to choose $c$
such that it is also not dependent on $k$ or $L$.
For this, we follow the argument of \citetext{gorecki:hhk:2018}.

The convergence rate of
$Y_{\bs} - \widetilde Y_{\bs}$
is
\begin{equation} \label{e:Ytildeconvergence}
\left| \widetilde{Y}_{m,\bs} - Y_{m,\bs} \right| = \cO_P( H_m(L) ),
\end{equation}
and it holds that
\begin{equation} \label{e:Ytildeconvergence2}
\var(Y_\bs - \widetilde Y_\bs) = \cO(H_m(L)^2).
\end{equation}
Since $Y_\bzero$ and $\widetilde Y_\bzero$ are Gaussian,
$ \var(Y_\bzero^k - \widetilde Y_\bzero^k) \xrightarrow{N \to \infty} 0$
follows easily.

\rightline{\QED}
\end{proof}

The next lemma shows that the statistics at each level $m$ are close to
their population counterparts.

\begin{lemma}\label{lemma:gaussianestimation}
Suppose $(X_\bs)_{\bs\in\mbZ^d}$ is a Gaussian process, which satisfies
Assumptions~\ref{ass:setup} and \ref{a:S3}, and Assumption~\ref{a:Rn}
holds. Then, the  following bounds hold as $N\to \infty$.
If $L \to \infty$, then
\begin{equation} \label{e:W1}
\cS_\bn - \widetilde \cS_\bn = o_P(1), \qquad
\cK_\bn - \widetilde \cK_\bn = o_P(1).
\end{equation}
If, in addition, Assumptions~\ref{ass:Fhat} and  \ref{ass:alphas} are satisfied,
then we can choose $L = L(N) \to \infty$ such that $L^d \, G(N) \convP 0$ and it holds that
\begin{equation} \label{e:W2}
\widehat \cS_\bn - \widetilde \cS_\bn = o_P(1), \qquad
\widehat \cK_\bn - \widetilde \cK_\bn = o_P(1).
\end{equation}
\end{lemma}
\begin{proof}
We begin by verifying the first relation in \refeq{W1}. Observe that
\begin{align*}
\cS_\bn - \widetilde \cS_\bn = & \;
\sqrt N \, ( \hat m^{Y}_3 - \hat m^{\widetilde Y}_3 )
- 3 \sqrt{N} \, \hat m_2^Y ( \hat m_1^{Y} - \hat m_1^{\widetilde Y} ) \\
 & \;
+ 3 \sqrt{N} \, \hat m^{\widetilde Y}_1 ( \hat m_2^{\widetilde Y} - \hat m_2^{Y} )
+ 2\sqrt{N} \left( (\hat m_1^Y)^3 - (\hat m_1^{\widetilde Y})^3 \right).
\end{align*}
By Lemma~\ref{lemma:convergence1},
the first three summands converge to zero in probability.
The last summand can be bounded by
$6\sqrt{N} \, | \hat m_1^Y - \hat m_1^{\widetilde Y} | \;
( |\hat m_1^Y| + |\hat m_1^{\widetilde Y}| )^2$,
which converges to zero in probability, too.

The second difference in \refeq{W1} can be handled similarly:
\begin{align*}
\cK_\bn - \widetilde \cK_\bn = & \;
\sqrt N \, ( \hat m^{Y}_4 - \hat m^{\widetilde Y}_4 )
- 4 \sqrt{N} \, \hat m_3^{Y} ( \hat m_1^{Y} - \hat m_1^{\widetilde Y} )
- 4 \sqrt{N} \, \hat m_1^{\widetilde Y} ( \hat m_3^{Y} - \hat m_3^{\widetilde Y} ) \\
& \;
+ 12 \sqrt{N} \, (\hat m_1^{Y})^2 ( \hat m_2^{Y} - \hat m_2^{\widetilde Y} )
+ 12 \sqrt{N} \, \hat m_2^{\widetilde Y} ( (\hat m_1^{Y})^2 - (\hat m_1^{\widetilde Y})^2 ) \\
& \;
- 3 \sqrt{N} \, (\hat m_2^{Y} - \hat m_2^{\widetilde Y})(\hat m_2^{Y} + \hat m_2^{\widetilde Y})
- 6 \sqrt{N} \, ( (\hat m_1^{Y})^4 - (\hat m_1^{\widetilde Y})^4 ).
\end{align*}
Here again, each summand converges to zero in probability.

Next we turn to the first relation in \refeq{W2}.
For this, we first recall the convergence rate for the estimators of the
filter functions.
Under Assumptions~\ref{ass:Fhat} and \ref{ass:alphas}, it holds that
\begin{equation} \label{e:h1}
\max_{\bl \in \mbZ^d} \norm{ \phi_{\bl} - \hat{\phi}_{\bl} }
= \cO(G(N)).
\end{equation}
For an exact proof of this, we refer to \citetext{kuenzer:hormann:kokoszka:2020}.
Now we can establish required  convergence rates of
$\hat m_k^{\widehat Y} - \hat m_k^{\widetilde Y}$ for $k \leq 3$.

For  $k = 1$,
\begin{align*}
\sqrt N ( \hat m_1^{\widehat Y} - \hat m_1^{\widetilde Y} ) = & \;
\frac{1}{\sqrt N} \sum_{\bs \in R_\bn} \sum_{\norm{\bl}_\infty \leq L}
\ip{X_{\bs-\bl}}{\hat \phi_{\bl} - \phi_{\bl}} \\
\E\norm{\sqrt N ( \hat m_1^{\widehat Y} - \hat m_1^{\widetilde Y} )}
 \leq & \;
\E\left[ \norm{\frac{1}{\sqrt N} \sum_{\bs \in R_\bn} X_{\bs}}^2 \right]^{1/2}
\sum_{\norm{\bl}_\infty \leq L} \E\left[ \norm{\hat \phi_{\bl} - \phi_{\bl}}^2 \right]^{1/2} \\
 \leq & \;
\left( \sum_{\bh \in \mbZ^d} \tr(C^X_\bh) \right)^{1/2} \!
\sum_{\norm{\bl}_\infty \leq L} \nu_2(\hat \phi_{\bl} - \phi_{\bl})
 = \cO( \E[ L^d \, G(N) ] ).
\end{align*}
By assumption, this converges to zero.

For  $k = 2$,
\begin{align*}
\left|
\hat m_2^{\widehat Y} - \hat m_2^{\widetilde Y}
\right|
 = & \;
\left|
\frac{1}{N} \sum_{\bs \in R_\bn}
\left( \widehat Y_\bs - \widetilde Y_\bs \right)
\left( \widehat Y_\bs + \widetilde Y_\bs \right)
\right| \\
 = & \;
\left|
\sum_{\norm{\bl}_\infty \leq L}
\ip{
\frac{1}{N} \sum_{\bs \in R_\bn}
( \widehat Y_\bs + \widetilde Y_\bs ) \,
X_{\bs-\bl}}{\hat \phi_{\bl} - \phi_{\bl}}
\right| \\
 \leq & \;
\sup_{\bl \in \mbZ^d} \norm{\hat \phi_{\bl} - \phi_{\bl}}
\sum_{\norm{\bl}_\infty \leq L}
\norm{
\frac{1}{N} \sum_{\bs \in R_\bn}
( \widehat Y_\bs + \widetilde Y_\bs ) \,
X_{\bs-\bl}}.
\end{align*}
It is easy to see that
\begin{align*}
\E\norm{\frac{1}{N} \sum_{\bs \in R_\bn}
( \widehat Y_\bs + \widetilde Y_\bs ) \,X_{\bs-\bl} } \leq & \;
\frac{1}{N} \! \sum_{\bs \in R_\bn}
\nu_2( \widehat Y_\bs + \widetilde Y_\bs ) \, \nu_2( X_{\bs-\bl} )\\
 &=
\Big( \sup_{\bs\in R_\bn} \nu_2( \widehat Y_\bs ) + \nu_2( \widetilde Y_\bzero ) \Big) \,
\nu_2( X_{\bzero} )
\end{align*}
This is uniformly bounded because
$L\to \infty$ and $L^d \, G(N) \convP 0$.
Therefore,
\begin{align*}
\left|
\hat m_2^{\widehat Y} - \hat m_2^{\widetilde Y}
\right|
 = & \;
 \cO_P(  L^d G(N) ),
\end{align*}
which converges to zero in probability.

For  $k = 3$,
\begin{align*}
\sqrt N ( \hat m_3^{\widehat Y} - \hat m_3^{\widetilde Y} ) = & \;
\frac{1}{\sqrt N} \sum_{\bs \in R_\bn}
\left( \widehat Y_\bs - \widetilde Y_\bs \right)^3
+ \frac{3}{\sqrt N} \sum_{\bs \in R_\bn}
\left( \widehat Y_\bs - \widetilde Y_\bs \right)^2 \widetilde Y_\bs \\
 & \;
+ \frac{3}{\sqrt N} \sum_{\bs \in R_\bn}
\left( \widehat Y_\bs - \widetilde Y_\bs \right) \widetilde Y_\bs^2.
\end{align*}
Following the earlier calculations,
we now only need to show uniform boundedness in probability
of the following sums:
\begin{align*}
\norm{\frac{1}{\sqrt N} \! \sum_{\bs \in R_\bn} \!
\widetilde Y_\bs^2 \, X_{\bs-\bl} },
& &
\hsnorm{\frac{1}{\sqrt N} \! \sum_{\bs \in R_\bn} \!
\widetilde Y_\bs \, X_{\bs-\bl} \otimes X_{\bs-\bk} },
& &
\hsnorm{\frac{1}{\sqrt N} \! \sum_{\bs \in R_\bn} \hspace{-3pt}
\left( \widehat Y_\bs - \widetilde Y_\bs \right) X_{\bs-\bl} \otimes X_{\bs-\bk} },
\end{align*}
where $\hsnorm{\cdot}$ denotes the Hilbert--Schmidt norm
defined by $\hsnorm{\Psi} = \iint \psi^2(u, v) du dv$.

The boundedness can be shown along the lines of Lemma~B.3 in
\citetext{gorecki:hhk:2018}.
All that is needed for this is stationarity and Isserlis's theorem.

We can now plug in these properties into
\begin{align*}
\widehat \cS_\bn - \widetilde \cS_\bn = & \;
\sqrt N \, ( \hat m^{\widehat Y}_3 - \hat m^{\widetilde Y}_3 )
- 3 \sqrt{N} \, \hat m_2^{\widehat Y} ( \hat m_1^{\widehat Y} - \hat m_1^{\widetilde Y} ) \\
 & \;
+ 3 \sqrt{N} \, \hat m^{\widetilde Y}_1 ( \hat m_2^{\widetilde Y} - \hat m_2^{\widehat Y} )
+ 2\sqrt{N} \left( (\hat m_1^{\widehat Y})^3 - (\hat m_1^{\widetilde Y})^3 \right)
\end{align*}
and see that it converges to zero in probability.
This concludes the verification of the first bound in \refeq{W2}.

It remains to verify the second bound in \refeq{W2}, i.e.
$\widehat \cK_\bn - \widetilde \cK_\bn = o_P(1)$. Observe that
\begin{align*}
\widehat \cK_\bn - \widetilde \cK_\bn = & \;
\sqrt N \, ( \hat m^{\widehat Y}_4 - \hat m^{\widetilde Y}_4 )
- 4 \sqrt{N} \, \hat m_3^{\widehat Y} ( \hat m_1^{\widehat Y} - \hat m_1^{\widetilde Y} )
- 4 \sqrt{N} \, \hat m_1^{\widetilde Y} ( \hat m_3^{\widehat Y} - \hat m_3^{\widetilde Y} ) \\
& \;
+ 12 \sqrt{N} \, (\hat m_1^{\widehat Y})^2 ( \hat m_2^{\widehat Y} - \hat m_2^{\widetilde Y} )
+ 12 \sqrt{N} \, \hat m_2^{\widetilde Y} ( (\hat m_1^{\widehat Y})^2 - (\hat m_1^{\widetilde Y})^2 ) \\
& \;
- 3 \sqrt{N} \, (\hat m_2^{\widehat Y} - \hat m_2^{\widetilde Y})(\hat m_2^{\widehat Y} + \hat m_2^{\widetilde Y})
- 6 \sqrt{N} \, ( (\hat m_1^{\widehat Y})^4 - (\hat m_1^{\widetilde Y})^4 ).
\end{align*}
We have shown that
\begin{align*}
\sqrt N ( \hat m_1^{\widehat Y} - \hat m_1^{\widetilde Y} ) & \convP 0, &
\hat m_2^{\widehat Y} - \hat m_2^{\widetilde Y} & \convP 0,  &
\sqrt N ( \hat m_3^{\widehat Y} - \hat m_3^{\widetilde Y} ) & \convP 0.
\end{align*}
This implies that
\begin{align*}
\widehat \cK_\bn - \widetilde \cK_\bn = & \;
\sqrt N \, \left(
\hat m^{\widehat Y}_4 - \hat m^{\widetilde Y}_4
- 3 (\hat m_2^{\widehat Y})^2 + 3 (\hat m_2^{\widetilde Y})^2
\right)
+ o_P(1).
\end{align*}
If we denote $\widehat D_\bs = \widehat Y_\bs - \widetilde Y_\bs$,
this term can be decomposed into the following terms
\begin{align}
& \frac{1}{\sqrt N} \sum_{\bs \in R_\bn}
\widehat D_\bs^2
\left(
\widehat D_\bs^2 - 3 \hat m_2^{\widehat D}
\right) \label{e:gaussianestd4} \\
& \frac{4}{\sqrt N} \sum_{\bs \in R_\bn}
\widehat D_\bs \widetilde Y_\bs
\left(
\widehat D_\bs^2 - 3 \hat m_2^{\widehat D}
\right) \label{e:gaussianestd3} \\
& \frac{6}{\sqrt N} \sum_{\bs \in R_\bn}
\left(
\widehat D_\bs^2 \widetilde Y_\bs^2
- \widehat D_\bs^2 \hat m_2^{\widetilde Y}
- 2 \widehat D_\bs \widetilde Y_\bs \hat m_1^{\widehat D \widetilde Y}
\right) \label{e:gaussianestd2} \\
& \frac{4}{\sqrt N} \sum_{\bs \in R_\bn}
\widehat D_\bs \widetilde Y_\bs
\left(
\widetilde Y_\bs^2 - 3 \hat m_2^{\widetilde Y}
\right) \label{e:gaussianestd1}
\end{align}
We will show the convergence to zero in probability
of the first term.
The others will follow by analogy.
The sums can be uniformly bounded in probability
by first isolating the estimated filter functions from the rest
and then taking the expected value of the remaining sums.

To concisely write out the proof, we will rely on tensor notation of higher order.
We rewrite \refeq{gaussianestd4} as follows, isolating the estimated
filter functions from the random field:
\begin{align}
\sum_{\bk_1, \bk_2, \bk_3, \bk_4} \hspace{-3pt}
\Big\langle A_{\bk_1,\bk_2,\bk_3,\bk_4}
\big(
(\hat \phi_{\bk_2} - \phi_{\bk_2}) \otimes
(\hat \phi_{\bk_3} - \phi_{\bk_3}) \otimes
(\hat \phi_{\bk_4} - \phi_{\bk_4}) \big), \,
(\hat \phi_{\bk_1} - \phi_{\bk_1}) \Big\rangle,
\label{e:gaussianestd4rewrite}
\end{align}
where the operator at the center of this expression is defined by
\begin{align*}
A_{\bk_1,\bk_2,\bk_3,\bk_4} =
\frac{1}{N^{3/2}} \sum_{\bs, \bt \in R_\bn} &
\Big(
X_{\bs - \bk_1} \otimes X_{\bs - \bk_2} \otimes X_{\bs - \bk_3} \otimes X_{\bs - \bk_4} \\ &
 - X_{\bs - \bk_1} \otimes X_{\bs - \bk_2} \otimes X_{\bt - \bk_3} \otimes X_{\bt - \bk_4} \\ &
 - X_{\bs - \bk_1} \otimes X_{\bt - \bk_2} \otimes X_{\bs - \bk_3} \otimes X_{\bt - \bk_4} \\ &
 - X_{\bs - \bk_1} \otimes X_{\bt - \bk_2} \otimes X_{\bt - \bk_3} \otimes X_{\bs - \bk_4}
\Big).
\end{align*}
From \refeq{h1}, it follows that
if $\norm{A_{\bk_1,\bk_2,\bk_3,\bk_4}} = \cO_P(1)$,
then \refeq{gaussianestd4} is $o_P(1)$.

Isserlis's theorem implies that for all jointly Gaussian elements
$U_1, U_2, U_3, U_4$,
we have
\begin{align*}
\E\left[U_1 \otimes U_2 \otimes U_3 \otimes U_4\right] =
C_{U_1,U_2} \otimes C_{U_3,U_4}
 + C_{U_1, U_3} \krpr C_{U_4, U_2}
 + C_{U_1, U_4} \krprt C_{U_3, U_2},
\end{align*}
where $ \widetilde \otimes\, $ denotes the Kronecker product
and $ \widetilde \otimes_\top\, $ denotes the transposed Kronecker product
on the operator space.
\[
( A \krpr B )\, C := A C B^*,
\qquad
( A \krprt B )\, C := A \bar C^* \bar B^*,
\]
On simple tensors, they rearrage the order
of the tensor product:
\[
(a \otimes b) \otimes (c \otimes d) =
(a \otimes c) \krpr (b \otimes d) =
(a \otimes d) \krprt (b \otimes c).
\]
For trace-class operators, $\tr(A \krpr B) = \tr(A) \tr(B)$.

We can now take the mean of the operator $A_{\bk_1,\bk_2,\bk_3,\bk_4}$.
Because of stationarity, the summation indices $\bs$ and $\bt$ can be replaced
by the spatial lag $\bh$.
All summands without this lag $\bh$ in the index cancel out.
\begin{align*}
\E[A_{\bk_1,\bk_2,\bk_3,\bk_4}] =
- \frac{1}{\sqrt N} \sum_{\bh} \prod_{i=1}^d \Big(1-\frac{|h_i|}{n_i} & \Big)
\Big( C_{\bk_3 - \bk_1 + \bh} \krpr C_{\bk_2 - \bk_4 - \bh}
 + C_{\bk_4 - \bk_1 + \bh} \krprt C_{\bk_2 - \bk_3 - \bh}
 \\ &
 + C_{\bk_2 - \bk_1 + \bh} \otimes C_{\bk_4 - \bk_3 + \bh}
 + C_{\bk_4 - \bk_1 + \bh} \krprt C_{\bk_2 - \bk_3 + \bh}
 \\ &
 + C_{\bk_2 - \bk_1 + \bh} \otimes C_{\bk_4 - \bk_3 - \bh}
 + C_{\bk_3 - \bk_1 + \bh} \krpr C_{\bk_2 - \bk_4 + \bh}
\Big).
\end{align*}
From the summability of the covariances, we then know that
$\norm{\E[A_{\bk_1,\bk_2,\bk_3,\bk_4}]} = \cO(N^{-1/2})$,
independently from $\bk_1, \bk_2, \bk_3, \bk_4$.

Using the same argument, we also show that
$\E\left[\norm{A_{\bk_1,\bk_2,\bk_3,\bk_4}}^2\right]$ is bounded.
Therefore, \refeq{gaussianestd4} is $o_P(1)$.

Conveniently, the convergence of \refeq{gaussianestd3}
follows easily from the argument on \refeq{gaussianestd4},
by noting that \refeq{gaussianestd3}
can be rewritten as
\[
4 \sum_{\bk_1, \bk_2, \bk_3, \bk_4} \hspace{-3pt}
\Big\langle A_{\bk_1,\bk_2,\bk_3,\bk_4}
\big(
(\hat \phi_{\bk_2} - \phi_{\bk_2}) \otimes
(\hat \phi_{\bk_3} - \phi_{\bk_3}) \otimes
(\hat \phi_{\bk_4} - \phi_{\bk_4}) \big), \,
\phi_{\bk_1} \Big\rangle,
\]
which has the same structure as \refeq{gaussianestd4rewrite}.
Because of the summability of the filter functions $\phi_{\bk_1}$,
\refeq{gaussianestd3} is then $o_P(1)$.
For \refeq{gaussianestd2} and \refeq{gaussianestd1},
this approach can be iterated.

\rightline{\QED}
\end{proof}

\bigskip

\noindent{\sc Proof of Theorem~\ref{t:main}:}
Because of Assumption~\ref{ass:Fhat},
we can estimate the spectral density in a way
such that $G(N) \convP 0$.
To achieve convergence of the estimated scores,
we choose the truncation index
$L = L(N) \xrightarrow{N\to\infty} \infty$
in a way such that
$L^d \, G(N) \convP  0$.
It is then clear that the estimators of the scores are consistent.

Under the assumption that $(X_\bs)$ is a Gaussian random field,
the population scores are also Gaussian.
This means that $Y_{m, \bs}$ and $Y_{m^\prime, \bt}$ (for $m \neq m^\prime$)
are not only uncorrelated but independent.
This independence also holds for the vectors
$( \cS_\bn^{(m)} \; \cK_\bn^{(m)} )^\top$.
Combining this property with Lemma~\ref{lemma:convergence1},
we see that
the vector $( \cS_\bn^{(1)} \; \cK_\bn^{(1)} \; \dots \;
\cS_\bn^{(p)} \; \cK_\bn^{(p)} )^\top$
converges in distribution to a vector of independent normal random variables.
Scaling with the asymptotic variances, it follows that
\[
T_p :=
\sum_{m=1}^p
\frac{ ( \cS_\bn^{(m)} )^2 }{ 6 \, \sigma_{\cS, m}^2 }
+
\frac{ ( \cK_\bn^{(m)} )^2 }{ 24 \, \sigma_{\cK, m}^2 }
\convd \chi^2_{2p}.
\]
Because Lemma~\ref{lemma:gaussianestimation} is also applicable,
$\widehat \cS_\bn^{(m)} - \cS_\bn^{(m)} \convP 0$ and
$\widehat \cK_\bn^{(m)} - \cK_\bn^{(m)} \convP 0$.
The estimators $\hat \sigma_{\cS, m}^2$ and $\hat \sigma_{\cK, m}^2$
of the variances are assumed to be consistent.
We can thus replace the population values with their estimators,
incurring only an asymptotically negligible error, i.e.
\[
| \widehat T_p - T_p | \convP 0
\]
Therefore, $\widehat T_p \convd \chi^2_{2p}$.

\rightline{\QED}

\bigskip

\noindent\textsc{Proof of Proposition~\ref{p:con}:}
We only consider  the case of $\hat\sigma_\cS^2$,
as the calculations for $\hat\sigma_\cK^2$ are similar.
Before we start the calculations,
we want to remark that in this proof
the estimation of the scores, the truncation of the summation
and the bias of the autocovariance estimators
need to be overcome at the same time.
This causes the calculations to be somewhat convoluted. Observe that
\begin{align*}
| \hat \sigma_\cS^2 - \sigma_\cS^2 | & \leq
\hspace{-3pt} \sum_{\norm{\bl}_\infty \leq L^\prime} \hspace{-3pt} | \hat \gamma_\bl^3 - \gamma_\bl^3 |
+ \hspace{-3pt} \sum_{\norm{\bl}_\infty > L^\prime} \hspace{-3pt} |\gamma_\bl|^3 \\
 & \leq
\hspace{-3pt} \sum_{\norm{\bl}_\infty \leq L^\prime} \hspace{-3pt}
| \hat \gamma_\bl - \gamma_\bl |^3
+ 3 \hspace{-3pt} \sum_{\norm{\bl}_\infty \leq L^\prime} \hspace{-3pt}
| \hat \gamma_\bl - \gamma_\bl |^2 |\gamma_\bl|
+ 3 \hspace{-3pt} \sum_{\norm{\bl}_\infty \leq L^\prime} \hspace{-3pt}
| \hat \gamma_\bl - \gamma_\bl | |\gamma_\bl|^2
+ \hspace{-3pt} \sum_{\norm{\bl}_\infty > L^\prime}
\hspace{-3pt} |\gamma_\bl|^3.
\end{align*}
The summability of $\gamma_\bl$ implies that the last sum vanishes if $L^\prime \to \infty$.
We now take the first sum as an example to show
how the convergence of each sum can be shown.
The calculations are similar to the previous proofs in this section,
and we will omit the terms $\hat m_1^{\widehat Y_m}$ and $\hat m_1^{Y_m}$ for notational simplicity,
as they are  asymptotically negligible.
Define
\[
\tilde c_\bl := \frac{1}{N} \hspace{-3pt} \sum_{\bs\in M_{\bh,\bn}} \hspace{-3pt}
\widetilde Y_{\bs+\bl} \widetilde Y_{\bs}
\qquad \text{ and } \qquad
\widehat D_\bs := \widehat Y_\bs - \widetilde Y_\bs.
\]
Then
\begin{align}
\sum_{\norm{\bl}_\infty \leq L^\prime} \hspace{-3pt} | \hat \gamma_\bl - \gamma_\bl |^3
 \leq & \;
\frac{a}{N^3} \hspace{-3pt} \sum_{\norm{\bl}_\infty \leq L^\prime} \hspace{-3pt}
\Big|
\sum_\bs \widehat D_{\bs+\bl} \widehat D_{\bs}
\Big|^3
+ \frac{a}{N^3} \hspace{-3pt} \sum_{\norm{\bl}_\infty \leq L^\prime} \hspace{-3pt}
\Big|
\sum_\bs \widehat D_{\bs+\bl} Y_{\bs}
\Big|^3 \label{e:cov3decomp}\\
&
+ \frac{a}{N^3} \hspace{-3pt} \sum_{\norm{\bl}_\infty \leq L^\prime} \hspace{-3pt}
\Big|
\sum_\bs Y_{\bs+\bl} \widehat D_{\bs}
\Big|^3
+ a \hspace{-3pt} \sum_{\norm{\bl}_\infty \leq L^\prime} \hspace{-3pt}
| \tilde c_\bl - \gamma_\bl |^3,
\nonumber
\end{align}
where $a$ is a fixed finite constant.
Looking at the first sum, we see that
\begin{align*}
&\frac{1}{N^3} \hspace{-3pt} \sum_{\norm{\bl}_\infty \leq L^\prime} \hspace{-3pt}
\Big|
\sum_\bs \widehat D_{\bs+\bl} \widehat D_{\bs}
\Big|^3\\
 & =
\frac{1}{N^3} \hspace{-3pt} \sum_{\norm{\bl}_\infty \leq L^\prime} \hspace{-2pt}
\Big| \hspace{-3pt}
\sum_{\norm{\bi}_\infty \leq L} \sum_{\norm{\bj}_\infty \leq L}
\ip{ (\sum_\bs X_{\bs+\bl-\bi} \otimes X_{\bs-\bj}) (\hat\phi_\bj-\phi_\bj)}{\hat\phi_\bi-\phi_\bi}
\Big|^3 \\
 & \leq
\frac{1}{N^3} \sup_{\bk \in \mbZ^d} \norm{\hat \phi_{\bk} - \phi_{\bk}}^6
\sum_{\norm{\bl}_\infty \leq L^\prime} \hspace{-3pt}
\Bigg(
\sum_{\norm{\bi}_\infty \leq L} \sum_{\norm{\bj}_\infty \leq L}
\Big\| \sum_\bs X_{\bs+\bl-\bi} \otimes X_{\bs-\bj} \Big\|
\Bigg)^3 \\
 & \leq
\Big( L^{d} \sup_{\bk \in \mbZ^d} \norm{\hat \phi_{\bk} - \phi_{\bk}} \Big)^6
\frac{1}{N^3 L^{2d}} \hspace{-3pt}
\sum_{\norm{\bl}_\infty \leq L^\prime} \hspace{-2pt}
\sum_{\norm{\bi}_\infty \leq L} \sum_{\norm{\bj}_\infty \leq L}
\Big\| \sum_\bs X_{\bs+\bl-\bi} \otimes X_{\bs-\bj} \Big\|^3.
\end{align*}
Here we used an inequality of the form
$(\sum_{k=1}^n a_k)^m \leq n^{m-1} \sum_{k=1}^n a_k^m$.
The first factor converges to zero in probability by assumption.
As for the single summands, stationarity implies that
\begin{align*}
&\E\left[
\hsnorm{ \sum_\bs X_{\bs+\bl} \otimes X_{\bs} }^4
\right]\\
 & =
\sum_{\bs_1} \sum_{\bs_2} \sum_{\bs_3} \sum_{\bs_4}
\E\left[
\ip{X_{\bs_1+\bl}}{X_{\bs_2+\bl}} \ip{X_{\bs_2}}{X_{\bs_1}}
\ip{X_{\bs_3+\bl}}{X_{\bs_4+\bl}} \ip{X_{\bs_4}}{X_{\bs_3}}
\right] \\
&  = \cO(N^2 + N^4 \tr(C_\bl)^4).
\end{align*}
It directly follows that
\begin{align*}
\E\Bigg[
\frac{1}{N^3 L^{2d}} \hspace{-3pt}
\sum_{\norm{\bl}_\infty \leq L^\prime} \hspace{-2pt}
\sum_{\norm{\bi}_\infty \leq L} \sum_{\norm{\bj}_\infty \leq L}
\Big\| \sum_\bs X_{\bs+\bl-\bi} \otimes X_{\bs-\bj} \Big\|^3
\Bigg]
 & = \cO((L^\prime)^d N^{-3/2} + \min(L^\prime/L, 1)^d),
\end{align*}
and in particular we can see that the expression is bounded.
Since we assume that $L^d G(N) \convP 0$,
it holds that
$\frac{1}{N^3} \hspace{-3pt} \sum_{\norm{\bl}_\infty \leq L^\prime} \hspace{-3pt}
\Big|
\sum_\bs \widehat D_{\bs+\bl} \widehat D_{\bs}
\Big|^3 \convP 0$.
The two other sums in \refeq{cov3decomp} can be bounded in a similar way.

What remains to show is that
$\sum_{\norm{\bl}_\infty \leq L^\prime} \hspace{-3pt}
| \tilde c_\bl - \gamma_\bl |^3 \convP 0$.
In the following calculations, we will use the notation
$c_\bl := \frac{1}{N} \hspace{-3pt} \sum_{\bs\in M_{\bh,\bn}} \hspace{-3pt}
Y_{\bs+\bl} Y_{\bs}$.
Note that these ``estimators" still exhibits a bias
because the sum divided by $N$ which is not the number of summands.
Because by assumption $L^\prime / \min n_i \to 0$, it follows that
the accumulated estimation bias converges to zero and
\begin{align*}
\E\left[
\sum_{\norm{\bl}_\infty \leq L^\prime} \hspace{-3pt} | c_\bl - \gamma_\bl |^3
\right]
\xrightarrow{N\to\infty} 0.
\end{align*}
Because we assume that $(L^{\prime})^{d/3} H_m(L) \to 0$,
it follows via Theorem~\ref{thm:isserlis} and
\refeq{Ytildeconvergence2} that
\begin{align*}
\E\left[
\sum_{\norm{\bl}_\infty \leq L^\prime} \hspace{-3pt} | c_\bl - \tilde c_\bl |^3
\right]
= \cO\Big( (L^\prime)^d \;
H_m(L)^3 \, \Big)
\xrightarrow{N\to\infty} 0.
\end{align*}

This shows that $| \hat \sigma_\cS^2 - \sigma_\cS^2 | \convP 0$.
The calculations for $\hat \sigma_\cK^2$ are analogous.
Note that since $(L^{\prime})^{d/4} \leq (L^{\prime})^{d/3}$,
the requirements on $L^\prime$ that stem from $\hat \sigma_\cK^2$
are fulfilled a fortiori.
\rightline{\QED}

\end{document}